\documentclass{article}
\usepackage{amsfonts}

\newtheorem{theorem}{Theorem}

\newtheorem{lemma}[theorem]{Lemma}

\newtheorem{proposition}[theorem]{Proposition}

\newenvironment{proof}[1][Proof]{\noindent \textbf{#1.} }{\  \rule{0.5em}{0.5em}}
\begin{document}

\title{On the Black's equation for the risk tolerance function\thanks{%
This work has been presented at seminars and workshops at Columbia, Oxford
and Standford. The authors would like to thank the participants for fruitful
comments and suggestions.}}
\author{Sigrid K\"allblad\thanks{%
TU-Wien; sigrid.kaellblad@tuwien.ac.at.} \ and Thaleia Zariphopoulou\thanks{%
Depts. of Mathematics and IROM, The University of Texas at Austin, and the
Oxford-Man Institute, University of Oxford; zariphop@math.utexas.edu.}}
\date{First draft: May 2016\thanks{%
Part of this work first appeared in the Ph.D. Thesis \cite{kallblad-thesis}
of the first author and in the preprint \cite{kallblad-Zar} of the authors. }%
; This draft: January 2017}
\maketitle

\begin{abstract}
We analyze a nonlinear equation proposed by F. Black (1968) for the optimal
portfolio function in a log-normal model. We cast it in terms of the risk
tolerance function and provide, for general utility functions, existence,
uniqueness and regularity results, and we also examine various monotonicity,
concavity/convexity and S-shape properties. Stronger results are derived for
utilities whose inverse marginal belongs to a class of completely monotonic
functions.
\end{abstract}

\section{Introduction}

In 1968\footnote{%
The authors would like to thank P. Carr for bringing this work to their
attention.}, Fisher Black discovered an autonomous equation, which he called
the "\textit{investment equation}", that the optimal portfolio function
satisfies, as a function of wealth, time and the optimal consumption in a
log-normal model (see \cite{black88}). With the notation used therein, the
equation is 
\begin{equation}
x_{2}=\left( r-c_{1}\right) x-\left( wr-c\right) x_{1}-\frac{1}{2}%
s^{2}x^{2}x_{11},  \label{Black-equation}
\end{equation}%
where $x$ is the optimal investment function (the subscripts $1,2$
correspond to partial derivatives in space and time) and $w$ the wealth
argument, $c$ is the optimal consumption function, and $r,s$ given market
parameters.

Black's equation was much later independently (re)discovered in \cite{he}
for a problem examining the compatibility of the optimal consumption
function and the utility function, and was also used in \cite{huang99} where
turnpike (long-term) properties of the optimal portfolio functions were
examined without intermediate consumption. More recently, it has been used
in \cite{agrawal}, \cite{bian}, \cite{fouque}, \cite{Shkolnikov}, \cite{Xia}%
, and others.

The aim herein is to provide a systematic study of the above equation in a
multi-stock log-normal model without intermediate consumption. We cast it in
terms of the local risk tolerance function $r\left( x,t\right) $ which, in a
trading horizon $\left[ 0,T\right] ,$ takes the form 
\[
r_{t}+\frac{1}{2}\left \vert \lambda \right \vert ^{2}r^{2}r_{xx}=0, 
\]%
$\left( x,t\right) \in \mathbb{R}_{+}\times \left[ 0,T\right) ],$ with $%
r\left( x,T\right) =R\left( x\right) ,$ the risk tolerance coefficient, and $%
\lambda $ being the market price of risk (cf. Proposition 3)\textbf{. }%
Throughout, we will refer to this nonlinear equation as the\textit{\ Black's
equation} for the risk tolerance function.

A related nonlinear equation is satisfied by the reciprocal $\gamma \left(
x,t\right) =\frac{1}{r\left( x,t\right) },$ known as the local risk aversion
function. It then follows that $\gamma \left( x,t\right) $ solves a \textit{%
porous medium} equation (PME),%
\[
\gamma _{t}-\frac{1}{2}\left( \gamma ^{-1}\right) _{xx}=0. 
\]%
The PME is typically classified by its "exponent" $m$ in the standard
representation $\gamma _{t}-\frac{1}{2}\left( \gamma ^{m}\right) _{xx}=0.$
In our case, $m=-1,$ which corresponds to the so-called \textit{fast
diffusion }regime $\left( m<1\right) $; see, for example, (see \cite{vasquez}%
).

Our contribution is multi-fold. We study the existence, uniqueness and
regularity of the solution to the Black's equation and provide estimates for
its derivatives. We also study its spatial monotonicity, concavity/convexity
and $S$-shaped properties. Specifically, we investigate when such properties
satisfied by the risk tolerance coefficient $R\left( x\right) $ are
inherited to $r\left( x,t\right) ,$ for all times $t\in \left[ 0,T\right) .$
We also study the time-monotonicity of $r\left( x,t\right) $ as well as its
dependence on $\left \vert \lambda \right \vert ^{2}.$ Furthermore, we
provide analogous results for the local relative risk tolerance function, $%
\tilde{r}\left( x,t\right) =\frac{r\left( x,t\right) }{x}.$

With the exception of the uniqueness result, for all other ones pivotal role
plays a function $H$ that solves the heat equation, appearing through the
nonlinear transformation, 
\begin{equation}
r\left( H\left( z,t\right) ,t\right) =H_{x}\left( z,t\right) ,
\label{r-H-preliminary}
\end{equation}%
$(z,t)\in \mathbb{R\times }\left[ 0,T\right] .$

This transformation results from a variation of the Legendre-Fenchel one
(applied to the value function), frequently used in stochastic optimization
with linear control dynamics. However, the form we propose is much more
convenient, for it simplifies a number of highly nonlinear expressions we
want to analyze. Indeed, these expressions reduce to much simpler ones
involving the partial derivatives of the aforementioned harmonic function.
Like $H$, these derivatives also solve the heat equation and, thus,
classical results can be in turn applied like the maximum principle, the
preservation of the log-concavity/convexity of its solutions, the properties
of their zero points sets, and others.

We note that this transformation was first proposed in an Ito-diffusion
setting in \cite{musiela-SIFIN} for a different class of risk preferences,
the so-called time-monotone forward performance processes. Therein, the
corresponding harmonic function satisfies the ill-posed heat equation, and
has quite different characteristics. Nevertheless, various expressions and
equations are algebraically similar but the nature of the solutions is
fundamentally different.

To show uniqueness, we do not use (\ref{r-H-preliminary}) but work directly
with a semilinear equation (cf. (\ref{eq-F})) satisfied by $r^{2}\left(
x,t\right) .$ In turn, the uniqueness in combination with the
convexity/concavity results yields the monotonicity of the risk tolerance
function with respect to $\left \vert \lambda \right \vert ^{2}.$

Some of the results we provide, namely, the monotonicity and the concavity,
have been shown before (see \cite{borell}) but we provide considerably
shorter and more direct alternative proofs, bypassing various lengthy
arguments. The uniqueness was also established in \cite{Xia} using duality
techniques, while we provide a much shorter proof based on PDE arguments.
The regularity results as well as the ones on the time monotonicity and the $%
S$-properties are, to the best of our knowledge, new.

Besides the results for general risk preferences, we examine the class of
utilities whose inverse marginal $I$ is a completely monotonic function of a
given form, namely, $I\left( x\right) =\int_{\alpha }^{\beta }x^{-y}d\mu
\left( y\right) ,$ $0<\alpha <\beta <\infty ,$ and $\mu $ a finite positive
Borel measure. For such cases, stronger bounds and regularity estimates can
be derived, which we provide.

We work in a multi-stock log-normal model. While this is a simple market
setting, we nevertheless gain a number of valuable insights for the risk
tolerance function and its derivatives that were not known before, and also
provide much shorter and direct proofs for existing results.

Beyond the log-normal model, our results can be also used in more general
diffusion settings for the analysis of the zeroth order term in stochastic
factor models with slow and fast factors (see, for example, \cite{agrawal}, 
\cite{fouque-hu}, \cite{fouque}, \cite{lorig}, \cite{lorig-sircar}). In
these models, this term is similar to the value function (\ref{value})
herein but with rescaled deterministic time.

The paper is organized as follows. In section 2 we introduce the investment
model and provide background and auxiliary results. In section 3 we derive
the risk tolerance equation, and study the uniqueness and regularity of its
solutions. In section 4, we provide further properties and conclude in
section 5 with the example of a family of completely monotonic inverse
marginals.

\section{The model and preliminary results}

\label{chap2model}

We start, for the reader's convenience, with a brief review of the classical
Merton problem (\cite{merton}). Investment takes place on $\left[ 0,T\right]
,$ a given trading horizon. The market environment consists of one riskless
and $N$ risky securities. The risky securities are stocks and their prices
are modelled as log-normal processes. Namely, for $i=1,...,N,$ the price $%
S_{t}^{i},$ $0\leq t\leq T,$ of the $i^{th}$ risky asset satisfies 
\begin{equation}
dS_{t}^{i}=S_{t}^{i}\left( \mu ^{i}dt+\sum_{j=1}^{N}\sigma
^{ji}dW_{t}^{j}\right) ,  \label{stock-price}
\end{equation}%
with $S_{0}^{i}>0.$ The process $W_{t}=\left( W_{t}^{1},...,W_{t}^{N}\right)
,$ $t\geq 0,$ is a standard $N-$dimensional Brownian motion defined on a
probability space $\left( \Omega ,\mathcal{F},\mathbb{P}\right) $ endowed
with the filtration $\mathcal{F}_{t}$ $=$ $\sigma \left( W_{s}:\textrm{ }0\leq
s\leq t\right) .$

The constants $\mu ^{i}$ and $\sigma ^{i}=\left( \sigma _{t}^{1i},...,\sigma
^{Ni}\right) ,i=1,...,N,$ $t\geq 0,$ take values in $\mathbb{R}$ and $%
\mathbb{R}^{N}$, respectively. For brevity, we use $\sigma $ to denote the $%
N\times N$ matrix volatility $\left( \sigma ^{ji}\right) ,$ whose $i^{th}$
column represents the volatility $\sigma ^{i}=\left( \sigma ^{1i},...,\sigma
^{Ni}\right) $ of the $i^{th}$ risky asset. Alternatively, we write equation
(\ref{stock-price}) as 
\[
dS_{t}^{i}=S_{t}^{i}\left( \mu ^{i}dt+\sigma ^{i}\cdot dW_{t}\right) . 
\]%
The riskless asset, the savings account, offers constant interest rate $r>0.$
We denote by $\mu $ the $N\times 1$ vector with coordinates $\mu ^{i}$ and
by $\mathbf{1}$ the $N-$dimensional vector with every component equal to
one. We assume that the volatility matrix is invertible, and define the
vector%
\begin{equation}
\lambda =\left( \sigma ^{T}\right) ^{-1}\left( \mu -r\mathbf{1}\right) .
\label{Lamda}
\end{equation}

Starting at $t\in \lbrack 0,T)$ with initial endowment $x>0$, the investor
invests at any time $s\in (t,T]$ in the riskless and risky assets. The
present value of the amounts invested are denoted, respectively, by $\pi
_{s}^{0}$ and $\pi _{s}^{i\textrm{ }}$, $i=1,...,N$, and are taken to be
self-financing. The (present) value of her investment is, then, given by $%
X_{s}^{\pi }=\sum_{k=0}^{N}\pi _{s}^{k},$ $s\in (t,T],$ which solves 
\begin{equation}
dX_{s}^{\pi }=\sigma \pi _{s}\cdot \left( \lambda ds+dW_{s}\right) ,
\label{discounted wealth}
\end{equation}%
with $X_{t}=x,$ and where the (column) vector, $\pi _{s}=\left( \pi _{s}^{i};%
\textrm{ }i=1,...,N\right) .$

A self-financing investment process $\pi _{s}$ is admissible if $\pi _{s}\in 
\mathcal{F}_{s},$

$E_{\mathbb{P}}\left( \int_{t}^{T}\left \vert \pi _{s}\right \vert
^{2}ds\right) <\infty $ and the associated wealth remains non-negative, $%
X_{s}^{\pi }\geq 0$, $\ s\in \left[ t,T\right] $. We denote the set of
admissible strategies by $\mathcal{A}$.

The utility function at $T$ is given by $U:\mathbb{R}_{+}\rightarrow \mathbb{%
R},$ and it is assumed to be a strictly concave, strictly increasing and $%
C^{\infty }\left( 0,\infty \right) $ function, satisfying the Inada
conditions $\lim_{x\downarrow 0}U^{\prime }(x)=\infty $ and $\lim_{x\uparrow
\infty }U^{\prime }(x)=0.$

We recall the inverse marginal $I:\mathbb{R}_{+}\rightarrow \mathbb{R}_{+},$ 
$I\left( x\right) =\left( U^{\prime }\right) ^{\left( -1\right) }\left(
x\right) $. It is assumed that it satisfies, for $C,\delta >0,$ 
\begin{equation}
I\left( x\right) \leq C\left( 1+x^{-\delta }\right) ,  \label{I-decay}
\end{equation}%
and for positive constants $c_{n},C_{n},$ $n=1,2,3,$ with $c_{2}>1,$ 
\begin{equation}
c_{1}I\left( x\right) \leq \left \vert xI^{\prime }\left( x\right) \right
\vert \leq C_{1}I\left( x\right) ,  \label{I-1}
\end{equation}%
\begin{equation}
c_{2}\left \vert I^{\prime }\left( x\right) \right \vert \leq xI^{\prime
\prime }\left( x\right) \leq C_{2}\left \vert I^{\prime }\left( x\right)
\right \vert \textrm{ \  \ and \  \ }\left \vert xI^{\prime \prime \prime
}\left( x\right) \right \vert \leq C_{3}I^{\prime \prime }\left( x\right) .%
\textrm{\ }  \label{I-2}
\end{equation}%
The above conditions are rather mild and satisfied by a large class of
utility functions. For example, if $I\left( x\right) =\Sigma
_{i=1}^{N}x^{-\alpha _{i}},$ $0<\alpha _{1}<...,\alpha _{N},$ (\ref{I-1})
and (\ref{I-2}) hold for $c_{1}=\alpha _{1},C_{1}=\alpha _{N},c_{2}=\alpha
_{1}+1,C_{2}=\alpha _{N}+1,C_{3}=\alpha _{N}+2.$

Throughout, we will use the domain notation $\mathbb{D}_{+}=\mathbb{R}%
_{+}\times \left[ 0,T\right] $ and $\mathbb{D}=\mathbb{R}\times \left[ 0,T%
\right] .$

The value function $u:\mathbb{D}_{+}\mathbb{\rightarrow R}_{+}$ is defined
as the maximal expected utility of terminal wealth,

\begin{equation}
u(x,t)=\sup_{\pi \in \mathcal{A}}E_{\mathbb{P}}\left( \left. U\left(
X_{T}^{\pi }\right) \right \vert X_{t}^{\pi }=x\right) ,  \label{value}
\end{equation}%
where $X_{s}^{\pi },$ $s\in \left( t,T\right] ,$ solves (\ref{discounted
wealth}). This optimization problem has been extensively studied (see, for
example, \cite{bjork09} and \cite{karatzas87}). It is known that $u\in
C^{\infty ,1}\left( \mathbb{D}_{+}\right) ,$ it is strictly increasing and
strictly concave in the spatial variable, and solves the
Hamilton-Jacobi-Bellman (HJB)\ equation,%
\begin{equation}
u_{t}-\frac{1}{2}\left \vert \lambda \right \vert ^{2}\frac{u_{x}^{2}}{u_{xx}%
}=0,  \label{HJB}
\end{equation}%
with $u(x,T)=U(x)$ and $\lambda $ as in (\ref{Lamda}).

The absolute\textit{\ }risk tolerance coefficient $R\left( x\right) $ and
the local risk tolerance function $r\left( x,t\right) $ are defined,
respectively, by%
\begin{equation}
R\left( x\right) =-\frac{U^{\prime }\left( x\right) }{U^{\prime \prime
}\left( x\right) }\textrm{ \  \  \  \  \ and \  \ }r\left( x,t\right) =-\frac{%
u_{x}(x,t)}{u_{xx}(x,t)}\textrm{,}  \label{r-def}
\end{equation}%
for $\left( x,t\right) \in \mathbb{D}_{+}.$ A standing assumption is that $%
R\left( 0\right) =0$ and that $R\left( x\right) $ is strictly increasing
(see \cite{arrow65}). The latter implies that $I^{\prime \prime }\left(
x\right) >0,$ $x>0.$

To ease the presentation, we eliminate the terminology "absolute" and
"local".

Next we introduce the function $H$ which plays a crucial role herein and
provide auxiliary results for it and its spatial derivatives.

\begin{lemma}
\label{helpus} Let $H:\mathbb{D}\rightarrow \mathbb{R}_{+}$ be defined by%
\begin{equation}
u_{x}\left( H(z,t),t\right) =e^{-z-\frac{1}{2}\left \vert \lambda \right
\vert ^{2}\left( T-t\right) },  \label{hh}
\end{equation}%
where $u\left( x,t\right) $ is as in (\ref{value}). Then, $H\left(
z,t\right) \in C^{\infty ,1}\left( \mathbb{D}\right) $ and solves the heat
equation 
\begin{equation}
H_{t}+\frac{1}{2}\left \vert \lambda \right \vert ^{2}H_{zz}=0,
\label{1heat}
\end{equation}%
with terminal condition%
\begin{equation}
H(z,T)=I\left( e^{-z}\right) .  \label{heat-terminal}
\end{equation}%
For each $t\in \left[ 0,T\right] ,$ it is strictly increasing and of full
range, 
\begin{equation}
\lim_{z\downarrow -\infty }H\left( z,t\right) =0\textrm{ \  \  \ and \  \ }%
\lim_{z\uparrow \infty }H\left( z,t\right) =\infty .\textrm{\ }
\label{H-fullrange}
\end{equation}%
Furthermore, the risk tolerance coefficient satisfies%
\begin{equation}
R\left( H\left( z,T\right) \right) =H_{z}\left( z,T\right) .
\label{risktolerance-T-H}
\end{equation}
\end{lemma}

\begin{proof}
The fact that $H$ is well defined follows from the spatial invertibility of $%
u.$ The regularity of $H$ and the fact that it solves the heat equation
follows directly from (\ref{HJB})\ and (\ref{hh}).

The existence and uniqueness of solutions to (\ref{1heat})\ follow from the
terminal datum (\ref{heat-terminal}) and property (\ref{I-decay}) which
yields $H\left( x,T\right) \leq C\left( 1+e^{\delta x}\right) $ (see, for
example, \cite{protter}, \cite{book}).

The monotonicity of $H$ follows from the strict spatial concavity of $u$. To
show (\ref{H-fullrange}), we use (\ref{hh}) and that the value function $%
u\left( x,t\right) $ satisfies, for $t\in \left[ 0,T\right] ,$ the Inada
conditions (see, for example, \cite{karatzas87}). Equality (\ref%
{risktolerance-T-H}) follows directly from (\ref{hh}) and (\ref{r-def}).
\end{proof}

\begin{lemma}
Assume inequalities (\ref{I-decay}), (\ref{I-1}) and (\ref{I-2}) hold, and
that $H$ solves (\ref{1heat}) and (\ref{heat-terminal}). Then, for $\left(
z,t\right) \in \mathbb{D}$ the following assertions hold for some positive
constants $M_{n},n_{n},N_{n}$, $n=1,2,3.$

i) The functions $\frac{\partial ^{n}H\left( z,t\right) }{\partial z^{n}}$
solve the heat equation (\ref{1heat}) with 
\begin{equation}
0<\frac{\partial ^{n}H\left( z,T\right) }{\partial z^{n}}\leq
M_{n}(1+e^{\delta z}),\textrm{ }n=1,2,  \label{H-terminals}
\end{equation}%
and%
\begin{equation}
\left \vert \frac{\partial ^{3}H\left( z,T\right) }{\partial z^{3}}\right
\vert \leq M_{3}(1+e^{\delta z}),  \label{H-terminal-third}
\end{equation}%
where $\delta $ as in (\ref{I-decay}).

ii) The following inequalities hold%
\begin{equation}
n_{n}\frac{\partial ^{n-1}H\left( z,t\right) }{\partial z^{n-1}}\leq \frac{%
\partial ^{n}H\left( z,t\right) }{\partial z^{n}}\leq N_{n}\frac{\partial
^{n-1}H\left( z,t\right) }{\partial z^{n-1}},\textrm{ }n=1,2,
\label{H-inequalities}
\end{equation}%
and 
\begin{equation}
\left \vert H_{zzz}\left( z,t\right) \right \vert \leq N_{3}\left \vert
H_{zz}\left( z,t\right) \right \vert .  \label{H-inequalities-third}
\end{equation}
\end{lemma}

\begin{proof}
i) The fact that the partial derivatives $\frac{\partial ^{n}H\left(
z,t\right) }{\partial z^{n}}$ solve the heat equation follows directly from (%
\ref{1heat}). To show (\ref{H-terminals}) recall that $H_{z}\left(
z,T\right) =-e^{-z}I^{\prime }\left( e^{-z}\right) >0$. Furthermore, (\ref%
{heat-terminal}) and (\ref{I-decay}) yield 
\[
H_{z}\left( z,T\right) =\left \vert e^{-z}I^{\prime }\left( e^{-z}\right)
\right \vert \leq C_{1}I\left( e^{-z}\right) \leq C_{1}C(1+e^{\delta z}), 
\]%
where we used (\ref{I-decay}) and (\ref{I-1}).

For $n=2,$ observe that (\ref{risktolerance-T-H}) yields $H_{z}\left(
z,T\right) R^{\prime }\left( H\left( z,T\right) \right) =H_{zz}\left(
z,T\right) .$ Using the assumption that $R^{\prime }>0$ and the full range
of $H,$ we deduce that $H_{zz}\left( z,T\right) >0.$ Furthermore, 
\[
H_{zz}\left( z,T\right) =e^{-z}I^{\prime }\left( e^{-z}\right)
+e^{-2z}I^{\prime \prime }\left( e^{-z}\right) <e^{-2z}I^{\prime \prime
}\left( e^{-z}\right) 
\]%
\[
\leq C_{2}\left \vert e^{-z}I^{\prime }\left( e^{-z}\right) \right \vert
\leq C_{2}C_{1}C(1+e^{\delta z}), 
\]%
where we used (\ref{I-1}), (\ref{I-2})\ and (\ref{I-decay}). Inequality (\ref%
{H-terminal-third}) follows similarly.

ii)\ Because of inequalities (\ref{H-terminals}) and (\ref{H-terminal-third}%
), we have that comparison holds for the heat equation satisfied by $%
H_{z},H_{zz}$ and $H_{zzz}$. Therefore, to show (\ref{H-inequalities}) and (%
\ref{H-inequalities-third}), it suffices to establish them for $t=T$ only,
which follows from direct differentiation of (\ref{heat-terminal}) and
repeated use of (\ref{I-1}) and (\ref{I-2}).
\end{proof}

\bigskip

We note that similar bounds to (\ref{H-terminal-third}) for higher order
partial derivatives, $\frac{\partial ^{n}H\left( z,T\right) }{\partial z^{n}}%
,$ $n>3,$ can be obtained if one imposes analogous to (\ref{I-1}) and (\ref%
{I-2}) inequalities for the partial derivatives $I^{\left( n\right) }\left(
x\right) $. Then, using that $\frac{\partial ^{n}H\left( z,t\right) }{%
\partial z^{n}}$ also satisfy the heat equation, we can deduce analogous to (%
\ref{H-inequalities-third}) bounds. These results can be strengthened if
further information about the sign of $\frac{\partial ^{n}H\left( z,T\right) 
}{\partial z^{n}}$ is known. We revert to such cases in section 5.

\section{Black's equation for the risk tolerance}

We start with the derivation of the Black's equation, and study questions on
existence, uniqueness and regularity of its solution.

\begin{proposition}
Let $r\left( x,t\right) $ be the risk tolerance function and $H\left(
x,t\right) $ the solution to the heat question (\ref{1heat}) and (\ref%
{heat-terminal}). Then, the following assertions hold.

i) For $\left( x,t\right) \in \mathbb{D}_{+},$ 
\begin{equation}
r\left( x,t\right) =H_{z}\left( H^{\left( -1\right) }\left( x,t\right)
,t\right) .  \label{r-H}
\end{equation}

ii) Furthermore, $r\left( x,t\right) \in C^{\infty ,1}\left( \mathbb{D}%
_{+}\right) $ and solves the nonlinear equation%
\begin{equation}
r_{t}+\frac{1}{2}\left \vert \lambda \right \vert ^{2}r^{2}r_{xx}=0\textrm{,}
\label{fastdiffusion}
\end{equation}%
with $r\left( x,T\right) =R\left( x\right) $ and $r\left( 0,t\right) =0,$ $%
t\in \left[ 0,T\right] .$
\end{proposition}

\begin{proof}
Property (\ref{r-H}) follows directly from (\ref{r-def}) and transformation (%
\ref{hh}). To show (\ref{fastdiffusion}), we have from (\ref{r-H})\ that 
\begin{equation}
r\left( H\left( z,t\right) ,t\right) =H_{z}\left( z,t\right) .
\label{r-H-auxiliary}
\end{equation}%
Differentiating twice yields, with all arguments of $H$ and its derivatives
evaluated at $\left( z,t\right) ,$%
\[
H_{t}r_{x}\left( H,t\right) +r_{t}\left( H,t\right) =H_{zt},\textrm{ \ }%
H_{z}r_{x}\left( H,t\right) =H_{zz} 
\]%
and 
\[
H_{zz}r_{x}\left( H,t\right) +H_{z}^{2}r_{xx}\left( H,t\right) =H_{zzz}. 
\]%
Using that $H_{z}$ solves (\ref{1heat}) we deduce 
\[
H_{t}r_{x}\left( H,t\right) +r_{t}\left( H,t\right) +\frac{1}{2}\left(
H_{zz}r_{x}\left( H,t\right) +H_{z}^{2}r_{xx}\left( H,t\right) \right) =0, 
\]%
and rearranging terms we obtain%
\[
r_{t}\left( H,t\right) +\frac{1}{2}\left \vert \lambda \right \vert
^{2}H_{z}^{2}r_{xx}\left( H,t\right) +r_{x}\left( H,t\right) \left( H_{t}+%
\frac{1}{2}\left \vert \lambda \right \vert ^{2}H_{zz}\right) =0. 
\]%
From (\ref{1heat}) we then get 
\[
r_{t}\left( H,t\right) +\frac{1}{2}\left \vert \lambda \right \vert
^{2}r^{2}\left( H,t\right) r_{xx}\left( H,t\right) =0, 
\]%
and using that, for each $t,$ the function $H\left( \cdot ,t\right) $ is of
full range we conclude. For the values of $r\left( 0,t\right) ,$ see \cite%
{Xia}.
\end{proof}

\bigskip

Next, we derive various regularity estimates on the derivatives of the risk
tolerance and of the ratio $r\left( x,t\right) /x.$ Other estimates may be
found in \cite{fouque-hu}.

\begin{proposition}
Assume that the inverse marginal utility function $I$ satisfies (\ref{I-1})
and (\ref{I-2}). Then, there exist positive constants $%
k_{n},l_{n},K_{n},L_{n},K,m,M,$ such that the following assertions hold.

i) For $n=0,1,$%
\begin{equation}
k_{n}\leq x^{n-1}\frac{\partial ^{n}r\left( x,t\right) }{\partial x^{n}}\leq
K_{n}\textrm{ \  \ and \ }\left \vert xr_{xx}\left( x,t\right) \right \vert
\leq K_{2},  \label{spatial-estimate}
\end{equation}%
and, for $n=0,1,2,$%
\begin{equation}
l_{n}\leq \left \vert x^{n}\frac{\partial ^{n}}{\partial x^{n}}\left( \frac{%
r\left( x,t\right) }{x}\right) \right \vert \leq L_{n}.
\label{Spatial-relative-estimate}
\end{equation}%
Furthermore, 
\begin{equation}
\left \vert r_{xx}^{2}(x,t)\right \vert \leq K.  \label{K-bound}
\end{equation}

ii)\ For each $x>0,$ $t\in \left[ 0,T\right) ,$%
\begin{equation}
mx\leq \left \vert r_{t}\left( x,t\right) \right \vert \leq Mx.
\label{r-time-derivative}
\end{equation}
\end{proposition}

\begin{proof}
i) We first show (\ref{spatial-estimate}) for $n=0.$ From (\ref%
{r-H-auxiliary}) and (\ref{H-inequalities}) we have 
\[
\frac{r\left( H\left( z,t\right) ,t\right) }{H\left( z,t\right) }=\frac{%
H_{z}\left( z,t\right) }{H\left( z,t\right) }\leq N_{1}, 
\]%
and the upper bound follows setting $K_{0}=N_{1}$. The lower bound follows
similarly.

Moreover, from (\ref{r-H}) and (\ref{H-inequalities}) we deduce 
\[
r_{x}\left( H\left( z,t\right) ,t\right) =\frac{H_{zz}\left( z,t\right) }{%
H_{z}\left( z,t\right) }\leq N_{2}, 
\]%
and thus the upper bound in (\ref{spatial-estimate})\ holds for $K_{1}=N_{2}$%
.

For $n=2,$ 
\begin{equation}
H_{z}\left( z,t\right) r_{xx}\left( H\left( z,t\right) ,t\right) =\frac{%
H_{zzz}\left( z,t\right) }{H_{z}\left( z,t\right) }-\left( \frac{%
H_{zz}\left( z,t\right) }{H_{z}\left( z,t\right) }\right) ^{2}
\label{r-xx-above}
\end{equation}%
and thus%
\[
\left \vert H\left( z,t\right) r_{xx}\left( H\left( z,t\right) ,t\right)
\right \vert \leq \frac{H\left( z,t\right) }{H_{z}\left( z,t\right) }\left(
\left \vert \frac{H_{zzz}\left( z,t\right) }{H_{z}\left( z,t\right) }\right
\vert +\left( \frac{H_{zz}\left( z,t\right) }{H_{z}\left( z,t\right) }%
\right) ^{2}\right) 
\]%
and the bound follows easily.

For (\ref{K-bound}), it suffices to observe, using (\ref{r-H-auxiliary}),
that%
\[
r_{xx}^{2}\left( H\left( z,t\right) ,t\right) =2\frac{H_{zzz}\left(
z,t\right) }{H_{z}\left( z,t\right) }, 
\]%
and we easily conclude using (\ref{H-inequalities}) and (\ref%
{H-inequalities-third}).

ii) From equation (\ref{fastdiffusion}) and equality (\ref{r-xx-above})
above, we have%
\[
r_{t}\left( H,t\right) =-\frac{1}{2}r\left( H,t\right) H_{z}r_{xx}\left(
H,t\right) 
\]%
\[
=-\frac{1}{2}H_{z}\left( \frac{H_{zzz}}{H_{z}}-\left( \frac{H_{zz}}{H_{z}}%
\right) ^{2}\right) . 
\]%
Thus,%
\[
\left \vert \frac{r_{t}\left( H,t\right) }{H}\right \vert \leq \frac{1}{2}%
\frac{H_{z}}{H}\left( \left \vert \frac{H_{zzz}}{H_{z}}\right \vert +\left( 
\frac{H_{zz}}{H_{z}}\right) ^{2}\right) , 
\]%
and the upper bound in (\ref{r-time-derivative}) follows.
\end{proof}

\subsection{Uniqueness of solutions}

\label{chap2robust}

Formula (\ref{r-H})\ yields the existence of smooth solutions to the risk
tolerance equation (\ref{fastdiffusion}). We next investigate the uniqueness.

To our knowledge, this question has been investigated only by Xia in \cite%
{Xia}, using an approximating sequence of penalized versions of equation (%
\ref{fastdiffusion}) and duality arguments to obtain comparison of their
solutions. As he mentions (see Remark 4.3 in \cite{Xia}), it is quite
difficult to obtain comparison results directly from the equation itself.

Herein, we provide such a result. The key idea is to consider an auxiliary
equation, specifically, the one satisfied by the square of the risk
tolerance function (cf. (\ref{eq-F})), and establish comparison for this
equation instead. The comparison result for (\ref{fastdiffusion}) would then
follow using the positivity of the risk tolerance functions.

\begin{proposition}
\label{propt3b} Let $I_{1},I_{2}$ be inverse marginal utility functions
satisfying (\ref{I-decay}),(\ref{I-1}) and (\ref{I-2}), and let $R_{1}$ and $%
R_{2}$ be the associated risk tolerance coefficients, satisfying, for $x\geq
0,$ 
\begin{equation}
R_{1}\left( x\right) \leq R_{2}\left( x\right) .  \label{R-inequality}
\end{equation}%
Then, for $\left( x,t\right) \in \mathbb{D}_{+},$ 
\begin{equation}
r_{1}\left( x,t\right) \leq r_{2}\left( x,t\right) \textrm{,}
\label{r-comparison}
\end{equation}%
with $r_{1},r_{2}$ solving (\ref{fastdiffusion}), with $r_{1}\left(
x,T\right) =R_{1}\left( x\right) ,r_{2}\left( x,T\right) =R_{2}\left(
x\right) .$
\end{proposition}

\begin{proof}
We first observe that equation (\ref{fastdiffusion}) yields that $F:=r^{2}$
solves the semilinear equation 
\begin{equation}
F_{t}+\frac{1}{2}\left \vert \lambda \right \vert ^{2}FF_{xx}-\frac{1}{4}%
\left \vert \lambda \right \vert ^{2}F_{x}^{2}=0,\textrm{ \ }  \label{eq-F}
\end{equation}%
$\left( x,t\right) \in \mathbb{D}_{+},$ with $F\left( x,T\right)
=R^{2}\left( x\right) $. To facilitate the exposition we will work with $%
\bar{F}\left( x,t\right) :=F\left( x,T-t\right) $ instead. Then, 
\begin{equation}
\bar{F}_{t}-\frac{1}{2}\left \vert \lambda \right \vert ^{2}\bar{F}\bar{F}%
_{xx}+\frac{1}{4}\left \vert \lambda \right \vert ^{2}\bar{F}_{x}^{2}=0,%
\textrm{ }  \label{eq-comparison}
\end{equation}%
with $\bar{F}\left( x,0\right) =R^{2}\left( x\right) .$ Therefore, the
functions $f\left( x,t\right) $ and $g\left( x,t\right) $ defined as%
\[
f\left( x,t\right) :=r_{1}^{2}\left( x,T-t\right) \textrm{ \  \  \ and \  \  \ }%
g\left( x,t\right) :=r_{2}^{2}\left( x,T-t\right) 
\]%
are (sub-) and (super-) solutions of (\ref{eq-comparison}), 
\[
f_{t}-\frac{1}{2}\left \vert \lambda \right \vert ^{2}ff_{xx}+\frac{1}{4}%
\left \vert \lambda \right \vert ^{2}f_{x}^{2}\leq 0\textrm{ \ and \ }g_{t}-%
\frac{1}{2}\left \vert \lambda \right \vert ^{2}gg_{xx}+\frac{1}{4}\left
\vert \lambda \right \vert ^{2}g_{x}^{2}\geq 0,\textrm{ } 
\]%
with 
\[
f\left( 0,t\right) =g\left( 0,t\right) =0\textrm{ \  \  \ and \  \  \ }f\left(
x,0\right) \leq g\left( x,0\right) . 
\]%
Furthermore, (\ref{K-bound}) yields, for $\left( x,t\right) \in \mathbb{D}%
_{+},$%
\begin{equation}
f_{xx}\left( x,t\right) \leq K\textrm{ \  \ and \ }g_{xx}\left( x,t\right) \leq
K.  \label{SSH-f-g}
\end{equation}%
We are going to establish that, for $t\in \left( 0,T\right) ,$ $x>0,$%
\begin{equation}
f\left( x,t\right) \leq g\left( x,t\right) .  \label{target}
\end{equation}%
For this, we follow parts of the proof of Theorem 3.1 in \cite{fukuda}. To
this end, let 
\begin{equation}
m=\left \{ 
\begin{array}{c}
T,\textrm{ \  \  \  \  \  \  \  \  \  \  \  \  \  \  \  \  \ if \ }T<\frac{1}{\left \vert
\lambda \right \vert ^{2}\left( 6K_{1}^{2}+\frac{1}{2}K\right) } \\ 
\\ 
m_{0}\in \left( 0,\frac{1}{\left \vert \lambda \right \vert ^{2}\left(
6K_{1}^{2}+\frac{1}{2}K\right) }\right) ,\textrm{ otherwise}%
\end{array}%
\right. ,  \label{m-parameter}
\end{equation}%
with $K_{1},K$ as in (\ref{spatial-estimate}) and (\ref{K-bound}),
respectively. Then, $m\leq T.$

Consider the domain $\mathbb{D}_{+}^{m}=\mathbb{R}_{+}\times \left[ 0,m%
\right] $ and introduce, for $\left( x,t\right) \in \mathbb{D}_{+}^{m},$ the
auxiliary functions%
\begin{equation}
f_{m}\left( x,t\right) :=\left( 1-\frac{t}{m}\right) f\left( x,t\right) \  \
\  \  \textrm{and\ }\  \ g_{m}\left( x,t\right) :=\left( 1-\frac{t}{m}\right)
g\left( x,t\right) .  \label{aux-functions}
\end{equation}%
As argued in \cite{fukuda}, a bootstrapping argumentation can be used to
establish (\ref{target}) once it is shown that, for $\left( x,t\right) \in 
\mathbb{D}_{+}^{m},$ the inequality 
\begin{equation}
f_{m}\left( x,t\right) \leq g_{m}\left( x,t\right)  \label{m-target}
\end{equation}%
holds. Next, observe that, for $\left( x,t\right) \in \mathbb{D}_{+}^{m},$
we have 
\begin{equation}
0\leq f_{m}\left( x,t\right) \leq f\left( x,t\right) \textrm{ \  \  \  \ and \  \
\ }0\leq g_{m}\left( x,t\right) \leq g\left( x,t\right) .  \label{f<fm}
\end{equation}%
Moreover,

\begin{equation}
f_{m}\left( x,0\right) \leq g_{m}\left( x,0\right) \textrm{ \  \ and \  \ }%
f_{m}\left( x,m\right) =g_{m}\left( x,m\right) =0,  \label{fm-gm-terminal}
\end{equation}%
and%
\begin{equation}
f_{m}\left( 0,t\right) =g_{m}\left( 0,t\right) =0.  \label{fm-gm-boundary}
\end{equation}%
\ By the assumption that $r_{1}\left( x,t\right) $ and $r_{2}\left(
x,t\right) $ are (sub- and super-) solutions of (\ref{fastdiffusion}), and
the definition of $f_{m}$ and $g_{m}$, we get%
\begin{equation}
\left( 1-\frac{t}{m}\right) f_{m,t}+\frac{1}{m}f_{m}-\frac{1}{2}\left \vert
\lambda \right \vert ^{2}f_{m}f_{m,xx}+\frac{1}{4}\left \vert \lambda \right
\vert ^{2}f_{m,x}^{2}\leq 0  \label{fm-sub}
\end{equation}%
and 
\begin{equation}
\left( 1-\frac{t}{m}\right) g_{m,t}+\frac{1}{m}g_{m}-\frac{1}{2}\left \vert
\lambda \right \vert ^{2}g_{m}g_{m,xx}+\frac{1}{4}\left \vert \lambda \right
\vert ^{2}g_{m,x}^{2}\geq 0.  \label{gm-super}
\end{equation}%
We now establish (\ref{m-target}). To this end, we consider the test
function $\varphi \left( x\right) =1+x^{4},$ $x\geq 0,$ and show that, for
any $\varepsilon >0,$%
\[
f_{m}\left( x,t\right) \leq g_{m}\left( x,t\right) +\varepsilon \varphi
\left( x\right) . 
\]%
We argue by contradiction, assuming that there exists $\varepsilon >0$ such
that 
\begin{equation}
\sup_{\left( x,t\right) \in \mathbb{D}_{+}}\left( f_{m}\left( x,t\right)
-g_{m}\left( x,t\right) -\varepsilon \varphi \left( x\right) \right) >0.
\label{auxuliary-inequality}
\end{equation}

Let $\left( x,t\right) $ be any point such that $f_{m}\left( x,t\right)
-g_{m}\left( x,t\right) -\varepsilon \varphi \left( x\right) >0.$ Using that 
$r_{i}\left( x,t\right) \leq K_{1}x$ (cf. (\ref{spatial-estimate})) and
inequalities (\ref{f<fm})\textbf{, }we deduce that $h\left( x,t\right) \leq
K_{1}^{2}x^{2},$\ for $h=f_{m},g_{m},$\ which together with the growth of $%
\varphi $ yields that $x<\infty .$ Furthermore, we observe that the extremum
in (\ref{auxuliary-inequality}), denoted by $\left( \tilde{x},\tilde{t}%
\right) ,$ is an interior point in $\mathbb{D}_{+}^{m}$. Indeed, if $\left( 
\tilde{x},\tilde{t}\right) $ is such that $\tilde{t}=0$ or $\tilde{t}=m$ we
get a contradiction from (\ref{fm-gm-terminal}), while if, for some $\tilde{t%
}\in \left( 0,m\right) ,$ $\left( \tilde{x},\tilde{t}\right) =\left( 0,%
\tilde{t}\right) ,$ we contradict (\ref{fm-gm-boundary}).

At the interior maximum $\left( \tilde{x},\tilde{t}\right) $ in (\ref%
{auxuliary-inequality}) we then have

\begin{equation}
f_{m}\left( \tilde{x},\tilde{t}\right) -g_{m}\left( \tilde{x},\tilde{t}%
\right) >\varepsilon \left( 1+\tilde{x}^{4}\right)  \label{max}
\end{equation}%
\begin{equation}
f_{m,t}\left( \tilde{x},\tilde{t}\right) -g_{m,t}\left( \tilde{x},\tilde{t}%
\right) =0,\textrm{ \  \ }f_{m,x}\left( \tilde{x},\tilde{t}\right)
-g_{m,x}\left( \tilde{x},\tilde{t}\right) =4\varepsilon \tilde{x}^{3}
\label{e-1}
\end{equation}%
and%
\begin{equation}
f_{m,xx}\left( \tilde{x},\tilde{t}\right) -g_{m,xx}\left( \tilde{x},\tilde{t}%
\right) \leq 12\varepsilon \tilde{x}^{2}.  \label{e-2}
\end{equation}%
\bigskip From (\ref{fm-sub}) and (\ref{gm-super}) we deduce%
\[
\frac{1}{m}\left( f_{m}\left( \tilde{x},\tilde{t}\right) -g_{m}\left( \tilde{%
x},\tilde{t}\right) \right) \leq \frac{1}{2}\left \vert \lambda \right \vert
^{2}\left( f_{m}\left( \tilde{x},\tilde{t}\right) f_{m,xx}\left( \tilde{x},%
\tilde{t}\right) -g_{m}\left( \tilde{x},\tilde{t}\right) g_{m,xx}\left( 
\tilde{x},\tilde{t}\right) \right) 
\]%
\[
-\frac{1}{4}\left \vert \lambda \right \vert ^{2}\left( f_{m,x}^{2}\left( 
\tilde{x},\tilde{t}\right) -g_{m,x}^{2}\left( \tilde{x},\tilde{t}\right)
\right) . 
\]%
In turn,%
\[
\left( \frac{1}{m}-\frac{1}{2}\left \vert \lambda \right \vert
^{2}f_{m,xx}\left( \tilde{x},\tilde{t}\right) \right) \left( f_{m}\left( 
\tilde{x},\tilde{t}\right) -g_{m}\left( \tilde{x},\tilde{t}\right) \right) 
\]%
\[
\leq \frac{1}{2}\left \vert \lambda \right \vert ^{2}g_{m}\left( \tilde{x},%
\tilde{t}\right) \left( f_{m,xx}\left( \tilde{x},\tilde{t}\right)
-g_{m,xx}\left( \tilde{x},\tilde{t}\right) \right) 
\]%
\[
-\frac{1}{4}\left \vert \lambda \right \vert ^{2}\left( f_{m,x}\left( \tilde{%
x},\tilde{t}\right) -g_{m,x}\left( \tilde{x},\tilde{t}\right) \right) \left(
f_{m,x}\left( \tilde{x},\tilde{t}\right) +g_{m,x}\left( \tilde{x},\tilde{t}%
\right) \right) . 
\]%
Inequality (\ref{K-bound}) and the definition of $f_{m}$ yield that $%
f_{m,xx}\left( x,t\right) <K.$ This, in combination with $f_{m}\left( \tilde{%
x},\tilde{t}\right) -g_{m}\left( \tilde{x},\tilde{t}\right) >0$ (cf. (\ref%
{max})) and (\ref{e-1}) and (\ref{e-2}), implies that 
\[
\left( \frac{1}{m}-\frac{1}{2}\left \vert \lambda \right \vert ^{2}K\right)
\left( f_{m}\left( \tilde{x},\tilde{t}\right) -g_{m}\left( \tilde{x},\tilde{t%
}\right) \right) 
\]%
\[
\leq 6\varepsilon \left \vert \lambda \right \vert ^{2}g_{m}\left( \tilde{x},%
\tilde{t}\right) \tilde{x}^{2}-\varepsilon \left \vert \lambda \right \vert
^{2}\tilde{x}^{3}\left( f_{m,x}\left( \tilde{x},\tilde{t}\right)
+g_{m,x}\left( \tilde{x},\tilde{t}\right) \right) . 
\]%
Using that the functions $r_{1}\left( x,t\right) $ and $r_{2}\left(
x,t\right) ,$ and in turn $f_{m}\left( x,t\right) $ and $g_{m}\left(
x,t\right) ,$ are strictly increasing, the above inequality yields that, at $%
\left( \tilde{x},\tilde{t}\right) ,$ we must have 
\[
\varepsilon \left( \frac{1}{m}-\frac{1}{2}\left \vert \lambda \right \vert
^{2}K\right) \left( 1+\tilde{x}^{4}\right) \leq 6\varepsilon \left \vert
\lambda \right \vert ^{2}g_{m}\left( \tilde{x},\tilde{t}\right) \tilde{x}%
^{2}. 
\]%
Finally, using once more that $g_{m}\left( \tilde{x},\tilde{t}\right) \leq
K_{1}^{2}\tilde{x}^{2}$, the above inequality then yields that $\frac{1}{m}-%
\frac{1}{2}\left \vert \lambda \right \vert ^{2}K\leq 6\left \vert \lambda
\right \vert ^{2}K_{1}^{2},$ and thus we must have 
\[
m\geq \frac{1}{\left \vert \lambda \right \vert ^{2}\left( 6K_{1}^{2}+\frac{1%
}{2}K\right) }, 
\]%
which, however, contradicts the choice of $m$ in (\ref{m-parameter}).
\end{proof}

\bigskip

We note that the property $r_{xx}^{2}\left( x,t\right) \leq K$ $(K>0),$
played a crucial role in the above proof. Such functions are frequently
called semi super-harmonic. \ 

We also note that in \cite{Xia} the admissible class of risk tolerance
functions satisfy $R\left( x\right) \leq M\left( 1+x\right) ,$ for $x\geq 0$
and $M>0.$ This property allows for risk tolerances with $\lim_{x\downarrow
0}R^{\prime }\left( x\right) =\infty ,$ which however are excluded herein.
This property makes our admissible class slightly smaller than the one in 
\cite{Xia}. Such a case is, for example, $R\left( x\right) =\sqrt{x}.$

\bigskip

A direct consequence of the above comparison result is the monotonicity of
the risk tolerance function on the market parameter $\left \vert \lambda
\right \vert ^{2}$. Note that while the terminal condition $r\left(
x,T\right) =R\left( x\right) $ is independent on $\left \vert \lambda \right
\vert ^{2},$ this is not the case for $t<T,$ since $r\left( x,t\right) =-$ $%
\frac{u_{x}\left( x,t\right) }{u_{xx}\left( x,t\right) },$ with both $%
u_{x}\left( x,t\right) $ and $u_{xx}\left( x,t\right) $ depending on $\left
\vert \lambda \right \vert ^{2}$ (cf. (\ref{HJB})).

As we demonstrate below, the dependence of $r\left( x,t\right) $ on $\left
\vert \lambda \right \vert ^{2}$ for \textit{all }times depends exclusively
on the \textit{curvature} of the \textit{terminal} condition $R\left(
x\right) .$

\begin{proposition}
Let the absolute risk tolerance coefficient $R\left( x\right) ,$ $x\geq 0,$
be concave (convex). Then, for $\left( x,t\right) \in \mathbb{D}_{+},$ the
risk tolerance function $r\left( x,t\right) $ is decreasing (increasing) in $%
\left \vert \lambda \right \vert ^{2}.$
\end{proposition}

\begin{proof}
Let $\left \vert \lambda \right \vert ^{2}>\left \vert \lambda ^{\prime
}\right \vert ^{2},$ and denote by $r\left( x,t;\lambda \right) $ and $%
r\left( x,t;\lambda ^{\prime }\right) $ the corresponding solutions to (\ref%
{fastdiffusion}) with $\left \vert \lambda \right \vert ^{2},\left \vert
\lambda ^{\prime }\right \vert ^{2}$ being used. We show that, for $\left(
x,t\right) \in \mathbb{D}_{+},$ $r\left( x,t;\lambda \right) \leq r\left(
x,t;\lambda ^{\prime }\right) .$

As it will be established in Proposition \ref{propmatrix}, if $R\left(
x\right) $ is concave, the risk tolerance function $r\left( x,t\right) $ is
also concave, for each $t\in \left[ 0,T\right) .$ Therefore,%
\[
r_{t}\left( x,t;\lambda ^{\prime }\right) +\frac{1}{2}\left \vert \lambda
\right \vert ^{2}r^{2}\left( x,t;\lambda ^{\prime }\right) r_{xx}\left(
x,t;\lambda ^{\prime }\right) 
\]%
\[
=r_{t}\left( x,t;\lambda ^{\prime }\right) +\frac{1}{2}\left \vert \lambda
^{\prime }\right \vert ^{2}r^{2}\left( x,t;\lambda ^{\prime }\right)
r_{xx}\left( x,t;\lambda ^{\prime }\right) 
\]%
\[
+\frac{1}{2}\left( \left \vert \lambda \right \vert ^{2}-\left \vert \lambda
^{\prime }\right \vert ^{2}\right) r^{2}\left( x,t;\lambda ^{\prime }\right)
r_{xx}\left( x,t;\lambda ^{\prime }\right) \leq 0, 
\]%
where we used that $r\left( x,t;\lambda ^{\prime }\right) $ solves (\ref%
{fastdiffusion}) and $r_{xx}\left( x,t;\lambda ^{\prime }\right) \leq 0.$

Therefore, $r\left( x,t;\lambda ^{\prime }\right) $ is a super-solution to
the equation that $r\left( x,t;\lambda \right) $ solves with terminal
condition $r\left( x,T;\lambda \right) =r\left( x,T;\lambda ^{\prime
}\right) ,$ and we easily conclude.
\end{proof}

\section{Properties of solutions of the Black's equation}

\label{chap2mon}

The previous results allow us to derive several properties of the risk
tolerance function $r\left( x,t\right) $ from analogous properties of the
risk tolerance coefficient $R\left( x\right) $. Some of these properties
have been studied in \cite{borell} and \cite{Xia}, but we provide
considerably shorter and more direct proofs. Furthermore, we establish
various new results.

We show that proving these properties amounts to only specifying the sign of
certain nonlinear quantities, which however reduce to much simpler
expressions involving the solutions to the heat equations and its spatial
derivatives. The latter also solve the heat equation, and thus we can, in
turn, use several results for this equation (comparison principle,
preservation of log-convexity/concavity, structure of their zero points
sets, etc.).

We start with the preservation of strict spatial monotonicity. This question
was examined by Arrow \cite{arrow65} who showed that, in a single period
problem with one risky stock, the optimal investment in the latter is
increasing in wealth if and only if the investor's utility exhibits
decreasing absolute risk aversion (\textit{DARA}), and as long as the risk
premium is positive. He, also, showed that the fraction of wealth invested
in the stock, known as the average propensity to invest, is decreasing in
wealth if and only if the utility exhibits increasing relative risk aversion
(\textit{IRRA}). Since this seminal work, it has become common in the
economic literature to frequently assume that the utility exhibits\textit{\
DARA} and \textit{IRRA}; these properties are also known as the \textit{%
Arrow hypothesis}. Similar results were later produced for discrete time
models (see, among others, \cite{Kimball-1}, \cite{lintner}, \cite{mossin}, 
\cite{pratt}, \cite{R-S1} and \cite{R-S2}, as well as \cite{eechodht} and 
\cite{Gollier-book} and references therein).

\begin{proposition}
\label{propt1}Let the absolute risk tolerance coefficient $R\left( x\right)
, $ $x\geq 0$, be strictly increasing. Then, for each $t\in \left[
0,T\right) , $ the risk tolerance function $r\left( x,t\right) $ is also
strictly increasing.
\end{proposition}

\begin{proof}
Differentiating (\ref{risktolerance-T-H}) for $t=T$ yields, 
\begin{equation}
R^{\prime }\left( x\right) =\left. \frac{H_{zz}\left( z,T\right) }{%
H_{z}\left( z,T\right) }\right \vert _{z=H^{\left( -1\right) }\left(
x,T\right) }.  \label{R'}
\end{equation}

Using the monotonicity of $H\left( z,t\right) $ and (\ref{H-fullrange}), we
see that $R\left( x\right) $ is strictly increasing if and only if the
auxiliary function $H\left( z,T\right) $ is strictly concave, $H_{zz}\left(
z,T\right) >0.$ On the other hand, the function $H_{zz}\left( z,t\right) $
also solves the heat equation (\ref{1heat}) with terminal condition $%
H_{zz}\left( z,T\right) =\left( I\left( e^{-z}\right) \right) ^{\prime
\prime }.$ From Proposition 2 and the comparison principle\textbf{\ }for (%
\ref{1heat}) we deduce that $H_{zz}\left( z,t\right) >0,$ $\left( z,t\right)
\in \mathbb{D}.$ Using (\ref{r-H}) we have that 
\[
r_{x}\left( x,t\right) =\left. \frac{H_{zz}\left( z,t\right) }{H_{z}\left(
z,t\right) }\right \vert _{z=H^{\left( -1\right) }\left( x,t\right) }, 
\]%
and using that $H_{z}\left( x,t\right) >0$ we conclude.
\end{proof}

\bigskip

Next, we examine how the convexity/concavity of risk tolerance coefficient $%
R\left( x\right) $ affects the behavior of the risk tolerance function $%
r\left( x,t\right) $ in \textit{both} space and time.

The curvature of $R\left( x\right) $ has been a topic of long-standing
debate. We refer the reader to \cite{gollier} for an extensive discussion as
well as to \cite{lajeri} and \cite{maggi}. As it is argued therein, there
are arguments and results which support both assumptions. Among others, it
is argued in \cite{hennessy06} that a concave risk tolerance coefficient
implies that the risk aversion is proper, standard and risk vulnerable (cf.,
respectively, \cite{zeck}, \cite{kimball} and \cite{gollier96}). An
empirical study in \cite{guiso08} also suggests that the risk tolerance is a
concave function of wealth.

For a portfolio problem with explicit solutions with a convex risk
tolerance, we refer the reader to \cite{chen11} and \cite{zariphopoulou09}.

\begin{proposition}
\label{propmatrix} Let the risk tolerance coefficient $R\left( x\right) ,$ $%
x\geq 0,$ be concave (convex). Then, the following assertions hold for the
risk tolerance function:

i)\ for each $x\geq 0,$ $r\left( x,t\right) $ is increasing (decreasing) in
time,

ii) for each $t\in \left[ 0,T\right) $, $r\left( x,t\right) $ is concave
(convex).
\end{proposition}

\begin{proof}
We first observe that if (ii) holds, then (i) follows directly from equation
(\ref{fastdiffusion}). To show the former, we argue as follows.
Differentiating (\ref{r-H}) twice at $t=T$ yields%
\begin{equation}
R^{\prime \prime }\left( x\right) =\frac{1}{H_{z}\left( z,T\right) }\left. 
\frac{\partial ^{2}}{\partial z^{2}}\left( \log H_{z}\left( z,T\right)
\right) \right \vert _{z=H^{\left( -1\right) }\left( x,T\right) }.
\label{better}
\end{equation}%
Therefore, $R\left( x\right) $ is concave if and only if $H_{z}\left(
z,T\right) $ is log-concave. But then, $H_{z}$ solves the heat equation (\ref%
{1heat}) with positive log-concave terminal data. Applying Proposition \ref%
{nylogconvex} for $h_{0}\left( x\right) =H_{z}\left( x,T\right) ,$ we deduce
that, for each $t\in \left[ 0,T\right) ,$ the function $H_{z}\left(
z,t\right) $ is also log-concave. Differentiating (\ref{r-H}) twice yields%
\[
r_{xx}\left( x,t\right) =\frac{1}{H_{z}\left( z,t\right) }\left. \frac{%
\partial ^{2}}{\partial z^{2}}\left( \log H_{z}\left( z,t\right) \right)
\right \vert _{_{z=H^{\left( -1\right) }\left( x,t\right) }}, 
\]%
and using (\ref{H-fullrange}) we conclude.

The convex case follows along similar arguments but using, instead, the
preservation of the log-convexity property of solutions to the heat equation
that $H_{z}\left( z,t\right) $ solves.
\end{proof}

\bigskip

For the case of concave risk tolerance coefficient, the above result is a
bit surprising if the investment horizon $\left[ 0,T\right] $ is long, for
it says that the investor should decrease his allocations (in terms of
feedback functions) in the risky assets as she gets older. This feature is
central in the management of life-cycle funds. We refer the reader, among
others, to \cite{bodie}, \cite{branch}, \cite{samuelson89}, \cite{spitzer}
and \cite{surz}. The temporal behavior of the value function and the optimal
policies have been examined in more extended model settings in \cite%
{choulli09} and \cite{larsen12}. Therein, however, the generality of the
model did not allow for specific results as the one above.

\bigskip

Next, we assume that the risk tolerance coefficient $R\left( x\right) $ is
an $S$-shaped function and examine whether this shape is propagated at
previous times. Key role plays the result of \cite{angenant} for the set of
zero points for an auxiliary linear parabolic equation (see (\ref{angenent})
below).

\begin{proposition}
Assume that there exists a unique $\hat{x}>0,$ such that the risk tolerance
coefficient $R\left( x\right) $ is convex (concave) in $\left[ 0,\hat{x}%
\right] $ and concave (convex) in $\left( \hat{x},\infty \right) $. Let $G:%
\mathbb{D}\rightarrow \mathbb{R}$ be defined as 
\begin{equation}
G\left( z,t\right) :=\frac{H_{zzz}\left( z,t\right) }{H_{zz}\left(
z,t\right) }-\frac{H_{zz}\left( z,t\right) }{H_{z}\left( z,t\right) },
\label{G-fn}
\end{equation}%
with $H$ solving (\ref{1heat}) and (\ref{heat-terminal}), and assume that,
for each $t\in \left[ 0,T\right] $, 
\begin{equation}
\lim_{z\downarrow -\infty }G\left( z,t\right) >0\textrm{ }\left( <0\right) 
\textrm{ \  \  \ and \ }\lim_{z\uparrow \infty }G\left( z,t\right) <0\textrm{ }(>0)
\label{G-1}
\end{equation}%
and 
\begin{equation}
G_{z}\left( \hat{z},T\right) \neq 0,\textrm{ \  \ for \ }\hat{z}=H^{\left(
-1\right) }\left( \hat{x},T\right) .  \label{G-2}
\end{equation}%
Then, there exists a unique continuous curve, $X\left( t\right) ,$ $t\in %
\left[ 0,T\right] ,$ with $X\left( T\right) =\hat{x},$ such that, for each $%
t\in \left[ 0,T\right) ,$ $r\left( x,t\right) $ is convex (concave) in $%
\left[ 0,X\left( t\right) \right] $ and concave (convex) in $\left( X\left(
t\right) ,\infty \right) .$
\end{proposition}

\begin{proof}
We only consider the case that $R\left( x\right) $ is convex in $\left[ 0,%
\hat{x}\right] $ and concave in $\left( \hat{x},\infty \right) ,$ since the
other case follows similarly. From (\ref{better}) we have 
\begin{equation}
R^{\prime \prime }\left( x\right) =\left. \frac{H_{zz}\left( z,T\right) }{%
H_{z}^{2}\left( z,T\right) }G\left( z,T\right) \right \vert _{z=H^{\left(
-1\right) }\left( x,T\right) }.  \label{better-1}
\end{equation}%
Using the assumptions for $R\left( x\right) $, the strict monotonicity and
full range of $H\left( x,T\right) ,$ we deduce that 
\begin{equation}
G\left( z,T\right) >0,\textrm{ }z<\hat{z},\textrm{ \ }G\left( \hat{z},T\right) =0%
\textrm{, \ }G\left( z,T\right) <0,\textrm{ }z>\hat{z}.  \label{G-again}
\end{equation}%
Next we show that there exists a unique continuous curve, say $Z\left(
t\right) ,$ $t\in \left[ 0,T\right] ,$ with $Z\left( T\right) =\hat{z},$
such that the above properties of $G\left( z,T\right) $ are "preserved" at
previous times, specifically, 
\begin{equation}
G\left( z,t\right) >0,\textrm{ }z<Z\left( t\right) ,\textrm{ \ }G\left( Z\left(
t\right) ,t\right) =0,\textrm{ \ }G\left( z,t\right) <0,\textrm{ }z>Z\left(
t\right) .\textrm{\ }  \label{G-good}
\end{equation}%
\ For this, define $u\left( z,t\right) :=\frac{H_{zzz}\left( z,t\right) }{%
H_{zz}\left( z,t\right) }$ and $v\left( z,t\right) :=\frac{H_{zz}\left(
z,t\right) }{H_{z}\left( z,t\right) },$ and observe that $G=u-v$ solves the
linear pde%
\begin{equation}
G_{t}\left( z,t\right) +\frac{1}{2}G_{zz}\left( z,t\right) +u\left(
z,t\right) G_{x}\left( z,t\right) +v_{z}\left( z,t\right) G\left( z,t\right)
=0,  \label{angenent}
\end{equation}%
with $G\left( z,T\right) $ as in (\ref{G-again}).

Using the estimates of the partial derivatives of $H$ we easily deduce that $%
G$ is smooth on $\mathbb{D}.$ This together with (\ref{G-2}) yield that in a
small neighborhood, say $B\left( \hat{z},t\right) ,$ $t\in \left(
T-\varepsilon ,T\right] ,$ $\varepsilon $ small, the ODE $\dot{Z}\left(
t\right) =-\frac{G_{t}\left( Z\left( t\right) ,t\right) }{G_{z}\left(
Z\left( t),t\right) \right) },$ with $Z\left( T\right) =\hat{z},$ has a
unique continuous solution $Z\left( t\right) $ on which $G\left( Z\left(
t\right) ,t\right) =0.$

Next, note that assumption (\ref{G-1})\ implies that for each $t\in \left[
0,T\right) ,$ there exists at least one zero point of $G\left( z,t\right) .$
On the other hand, the results of \cite{angenant} for the set of zero points
for solutions of linear PDE, as the equation (\ref{angenent}) above, yield
that the set of these points cannot be decreasing as time increases. In
other words, the number of zero points of $G\left( z,t_{1}\right) $ is less
or equal the number of zero points of $G\left( z,t_{2}\right) ,$ for $0\leq
t_{1}<t_{2}.$ On the other hand, there is a unique curve of zero points in $%
\left( T-\varepsilon ,T\right] ,$ and thus for every $t\in \left[
0,T-\varepsilon \right] ,$ there must be also a unique zero point, denoted
by $Z\left( t\right) .$

To show the continuity of $Z\left( t\right) ,$ $t\in \left[ 0,T\right] ,$ we
argue as follows. Let $t_{0}\in \left[ 0,T\right) $ and consider a sequence $%
t_{n}\rightarrow t_{0},$ with $G\left( Z\left( t_{n}\right) ,t_{n}\right)
=0. $ Assumption (\ref{G-1}) implies that $Z\left( t_{n}\right) $ is
bounded, and therefore along a subsequence, say $Z\left( t_{k_{n}}\right) ,$
we have that $Z\left( t_{k_{n}}\right) \rightarrow z_{0},$ for some $%
z_{0}\in \mathbb{R}.$ Moreover, $G\left( Z\left( t_{k_{n}}\right)
,t_{k_{n}}\right) \rightarrow G\left( z_{0},t_{0}\right) $ and thus $G\left(
z_{0},t_{0}\right) =0.$ However, for each $t_{0},$ there is a unique zero
point and, thus, it must be that $z_{0}=Z\left( t_{0}\right) .$ Therefore,
there is a unique continuous curve $Z\left( t\right) ,$ $t\in \left[ 0,T%
\right] ,$ such that (\ref{G-good}) holds.

Recall that, similarly to (\ref{better-1}), we have that 
\begin{equation}
r_{xx}\left( x,t\right) =\left. \frac{H_{zz}\left( z,t\right) }{%
H_{z}^{2}\left( z,t\right) }G\left( z,t\right) \right \vert _{z=H^{\left(
-1\right) }\left( x,t\right) },  \label{r-again}
\end{equation}%
with $\frac{H_{zz}\left( z,t\right) }{H_{z}^{2}\left( z,t\right) }>0,$ $%
\left( z,t\right) \in \mathbb{D}.$ 

Next, we define, for $t\in \left[ 0,T\right] ,$ the curve 
\[
X\left( t\right) :=H\left( Z\left( t\right) ,t\right) .
\]%
Let $t_{0}$ in $\left[ 0,T\right) $ and $x_{0}<X\left( t_{0}\right) .$ The
strict spatial monotonicity of $H^{\left( -1\right) }$ yields that $%
H^{\left( -1\right) }\left( x,t_{0}\right) <Z\left( t_{0}\right) $ and, in
turn, (\ref{G-good}) implies that $G\left( H^{\left( -1\right) }\left(
x_{0},t_{0}\right) ,t_{0}\right) >0.$ Then, (\ref{r-again}) gives $%
r_{xx}\left( x_{0},t_{0}\right) >0.$ The rest of the proof follows easily.
\end{proof}

\subsection{The relative risk tolerance}

The relative risk tolerance coefficient $\tilde{R}\left( x\right) $ and the
relative risk tolerance function $\tilde{r}\left( x,t\right) $ are defined,
for $x>0,$ $t\in \left[ 0,T\right] ,$ as%
\[
\tilde{R}\left( x\right) =-\frac{U^{\prime }\left( x\right) }{xU^{\prime
\prime }\left( x\right) }\textrm{ \  \  \ and \  \ }\tilde{r}\left( x,t\right) =-%
\frac{u_{x}\left( x,t\right) }{xu_{xx}\left( x,t\right) }.
\]%
There are also given in the implicit form%
\begin{equation}
\tilde{R}\left( H\left( z,T\right) \right) =\frac{\partial }{\partial z}\log
H\left( z,T\right) \textrm{ \  \  \ and \  \ }\tilde{r}\left( H\left( z,t\right)
,t\right) =\frac{\partial }{\partial z}\log H\left( z,t\right) ,
\label{relative-aux}
\end{equation}%
as it follows from (\ref{r-H}). It also follows that $\tilde{r}\left(
x,t\right) $ solves the \textit{Burger's }equation%
\[
\tilde{r}_{t}\left( x,t\right) +\frac{1}{2}\left \vert \lambda \right \vert
^{2}\tilde{r}_{xx}\left( x,t\right) +\left \vert \lambda \right \vert ^{2}%
\tilde{r}\left( x,t\right) \tilde{r}_{x}\left( x,t\right) =0,
\]%
with $\tilde{r}\left( x,T\right) =\tilde{R}\left( x\right) .$

Representation (\ref{relative-aux}) can be used to prove the following
results. The first was proved by \cite{Xia}, and we provide a much shorter
proof. The second result is, to the best of our knowledge, new.

\begin{proposition}
Assume that the relative risk tolerance coefficient $\tilde{R}\left(
x\right) $ is increasing (decreasing). Then, for each $t\in \left(
0,T\right) ,$ the relative risk tolerance function $\tilde{r}\left(
x,t\right) $ is also increasing (decreasing).
\end{proposition}

\begin{proof}
Differentiating (\ref{r-H}) at $t=T$ yields$,$ $\tilde{R}^{\prime }\left(
x\right) =\left. \frac{\partial ^{2}}{\partial z^{2}}\left( \log H\left(
z,T\right) \right) \right \vert _{z=H^{\left( -1\right) }\left( x,T\right)
}. $ Therefore, $\tilde{R}\left( x\right) $ is an increasing function if and
only if $H\left( z,T\right) $ is log-convex. Applying Proposition \ref%
{nylogconvex} we deduce that, for each $t\in \left[ 0,T\right) $, $H\left(
z,t\right) $ is log-convex. Using (\ref{r-H}) once more, we have that 
\[
\frac{\partial }{\partial x}\left( \tilde{r}\left( x,t\right) \right)
=\left. \frac{\partial ^{2}}{\partial z^{2}}\left( \log H\left( z,t\right)
\right) \right \vert _{z=H^{\left( -1\right) }\left( x,t\right) }, 
\]%
and using (\ref{H-fullrange}) we conclude.

The analogous results are readily derived when the relative risk tolerance
coefficient is decreasing. In this case, one uses the preservation of the
log-concavity property of the solution to the heat equation.
\end{proof}

\bigskip

The following result is related to the concavity/convexity of the relative
risk tolerance function.

\begin{proposition}
Assume that the function $H,$ solving (\ref{1heat}) and (\ref{heat-terminal}%
), satisfies, for $z\in \mathbb{R},$%
\begin{equation}
\left( \log H\left( z,T\right) \right) _{zz}<0\textrm{ \  \ and \  \ }\left(
\log H\left( z,T\right) \right) _{zzz}>0,  \label{condition}
\end{equation}%
(resp. $\left( \log H\left( z,T\right) \right) _{zz}>0$ and $\left( \log
H\left( z,T\right) \right) _{zzz}<0).$ Then, the relative risk tolerance
coefficient satisfies $\tilde{R}^{\prime \prime }\left( x\right) >0$ (resp. $%
\tilde{R}^{\prime \prime }\left( x\right) <0$) and, for $t\in \left[
0,T\right) ,$ the relative risk tolerance function is convex (concave) in
the spatial variable.
\end{proposition}

\begin{proof}
First, observe that 
\begin{equation}
\tilde{R}^{\prime \prime }\left( H\left( z,T\right) \right) =\frac{1}{%
H_{z}\left( z,T\right) }\left( \frac{\left( \log H\left( z,T\right) \right)
_{zz}}{H_{z}\left( z,T\right) }\right) _{z}  \label{R''}
\end{equation}%
\[
=\frac{1}{H_{z}\left( z,T\right) }\frac{\left( \log H\left( z,T\right)
\right) _{zzz}H_{z}\left( z,T\right) -\left( \log H\left( z,T\right) \right)
_{zz}H_{zz}\left( z,T\right) }{H_{z}^{2}\left( z,T\right) }. 
\]%
Therefore, if assumption (\ref{condition}) holds, we have that $\tilde{R}%
^{\prime \prime }\left( x\right) >0.$

Next, we claim that, for $\left( z,t\right) \in \mathbb{D},$ $\left( \log
H\left( z,t\right) \right) _{zz}<0$ \ and $\left( \log H\left( z,t\right)
\right) _{zzz}>0.$

For simplicity, we take $\left \vert \lambda \right \vert ^{2}=1.$ The first
assertion is proved in the Appendix. For the second inequality, we introduce
the functions $g,F:\mathbb{D\rightarrow R}$, defined as $g\left( z,t\right)
=\log H\left( z,t\right) $\ and $F\left( z,t\right) =g_{zzz}\left(
z,t\right) .$ Then, $F$ satisfies%
\[
F_{t}\left( z,t\right) +\frac{1}{2}F_{zz}\left( z,t\right) +g_{z}\left(
z,t\right) F_{z}\left( z,t\right) +3g_{zz}\left( z,t\right) F\left(
z,t\right) =0, 
\]%
with $F\left( z,T\right) >0.$

Furthermore, $g_{zz}\left( z,t\right) <0,$ and $F\left( z,T\right) \leq k,$
for some positive constant $k.$ The latter follows for direct
differentiation and repeated use of (\ref{H-inequalities}) and (\ref%
{H-inequalities-third}). Therefore, applying the comparison principle for
the above linear equation yields that $F\left( z,t\right) >0.$ Using that 
\[
\tilde{r}_{xx}\left( H\left( z,t\right) ,t\right) =\frac{1}{H_{z}\left(
z,t\right) }\frac{F\left( z,t\right) H_{z}\left( z,t\right) -g_{zz}\left(
z,t\right) H_{zz}\left( z,t\right) }{H_{z}^{2}\left( z,t\right) } 
\]%
we obtain that $\tilde{r}_{xx}\left( H\left( z,t\right) ,t\right) >0$ and we
conclude.

The second case, $\left( \log H\left( z,T\right) \right) _{zz}>0$ and $%
\left( \log H\left( z,T\right) \right) _{zzz}<0,$ is more involved, because
the coefficient of the zeroth-order term is now positive$,$ $g_{zz}\left(
z,t\right) >0,$ as it follows from the assumption that $\left( \log H\left(
z,T\right) \right) _{zz}>0$ and Proposition 13. Note, however, that 
\[
\left \vert \left( \log H\left( z,t\right) \right) _{zz}\right \vert =\left
\vert \frac{H_{zzz}\left( z,t\right) }{H\left( z,t\right) }-\frac{%
H_{zz}\left( z,t\right) }{H\left( z,t\right) }\frac{H_{z}\left( z,t\right) }{%
H\left( z,t\right) }\right. 
\]%
\[
\left. -2\frac{H_{z}\left( z,t\right) }{H\left( z,t\right) }\left( \frac{%
H_{zz}\left( z,t\right) }{H\left( z,t\right) }-\left( \frac{H_{z}\left(
z,t\right) }{H\left( z,t\right) }\right) ^{2}\right) \right \vert , 
\]%
and we conclude using (\ref{H-inequalities}) and (\ref{H-inequalities-third}%
). The rest of the proof follows easily.
\end{proof}

\section{A class of completely monotonic inverse marginals}

We consider utility functions whose inverse marginal $I$ is of the form 
\begin{equation}
I\left( x\right) =\int_{a}^{b}x^{-y}\nu \left( dy\right) ,  \label{I-cm}
\end{equation}%
where $\nu $ is a positive finite Borel measure, with support $1<a<b<\infty
. $ This is a \textit{completely monotonic} function, since $\left(
-1\right) ^{n}I^{\left( n+1\right) }\left( x\right) <0,$ $n=1,2,...,$ $x\in 
\mathbb{R}_{+}.$

\bigskip

\textit{Example}: i) $I\left( x\right) =x^{-\frac{1}{1-\gamma }},$ $\frac{1}{%
1-\gamma }\in \left[ a,b\right] ,$ which corresponds to the marginal utility 
$U^{\prime }\left( x\right) =x^{1-\gamma }.$

ii) $I\left( x\right) =x^{-\frac{1}{1-\gamma }}+x^{-\frac{2}{1-\gamma }},$ $%
\frac{1}{1-\gamma }\in \left[ a,\frac{b}{2}\right] ,$ which corresponds to
marginal utility $U^{\prime }\left( x\right) =2^{1-\gamma }\left( \sqrt{1+4x}%
-1\right) ^{\gamma -1}.$

\bigskip

We then see that (\ref{heat-terminal})\ yields that $H\left( x,T\right) $ is
an \textit{absolutely monotonic function, given by}%
\begin{equation}
H\left( z,T\right) =\int_{a}^{b}e^{zy}\nu \left( dy\right) .  \label{H-am}
\end{equation}%
It follows easily that, for $n=1,2,...,$ 
\[
\left \vert xI^{(n)}\left( x\right) \right \vert
=x\int_{a}^{b}y^{n}x^{-y-n}\nu \left( dy\right) \leq
b\int_{a}^{b}y^{n-1}x^{-y-\left( n-1\right) }\nu \left( dy\right)
=b\left \vert I^{\left( n-1\right) }\left( x\right) \right \vert ,
\]%
and, turn, 
\[
a\left \vert I^{\left( n-1\right) }\left( x\right) \right \vert \leq
\left \vert xI^{(n)}\left( x\right) \right \vert \leq b\left \vert I^{\left(
n-1\right) }\left( x\right) \right \vert .
\]%
Furthermore, at $t=T,$ $n=1,2,...,$ 
\begin{equation}
\frac{\partial ^{n}}{\partial z^{n}}H\left( z,T\right)
=\int_{a}^{b}y^{n}e^{zy}\nu \left( dy\right) >0,  \label{H-cm}
\end{equation}%
and, thus,%
\begin{equation}
a\frac{\partial ^{n-1}}{\partial z^{n-1}}H\left( z,T\right) \leq \  \frac{%
\partial ^{n}}{\partial z^{n}}H\left( z,T\right) \leq b\frac{\partial ^{n-1}%
}{\partial z^{n-1}}H\left( z,T\right) .  \label{H-cm-estimates}
\end{equation}%
The above together with the comparison principle for the heat equation
(which follows from (\ref{H-am}) and (\ref{I-decay})) yield that, for $%
\left( z,t\right) \in \mathbb{D},$ 
\begin{equation}
a\frac{\partial ^{n-1}}{\partial z^{n-1}}H\left( z,t\right) \leq \  \frac{%
\partial ^{n}}{\partial z^{n}}H\left( z,t\right) \leq b\frac{\partial ^{n-1}%
}{\partial z^{n-1}}H\left( z,t\right) .  \label{H-cm-estimates-t}
\end{equation}

\begin{proposition}
Assume that the inverse marginal has the form (\ref{I-cm}). Then, the
following assertions hold:

i) For $\left( x,t\right) \in \mathbb{D}_{+},$ 
\[
ax\leq r\left( x,t\right) \leq bx. 
\]

ii)\ The risk tolerance function is convex in $x$ and decreasing in time,
for $\left( x,t\right) \in \mathbb{D}_{+}.$

iii) For $n=1,2,...,$ there exist positive constants $K_{n},L_{n},$ such that%
\begin{equation}
\left \vert x^{n-1}\frac{\partial ^{n}r\left( x,t\right) }{\partial x^{n}}%
\right \vert \leq K_{n}\textrm{ \  \  \ and \  \ }\left \vert x^{n}\frac{\partial
^{n}}{\partial x^{n}}\left( \frac{r\left( x,t\right) }{x}\right) \right
\vert \leq L_{n}.  \label{estimates-cm}
\end{equation}
\end{proposition}

\begin{proof}
Part (i) follows directly from (\ref{r-H}) and (\ref{H-cm}).

To show (ii), we first observe that $H_{z}\left( z,T\right) H_{zzz}\left(
z,T\right) -H_{zz}^{2}\left( z,T\right) \geq 0,$ as it follows from the
inequality%
\[
\int_{a}^{b}ye^{zy}\nu \left( dy\right) \int_{a}^{b}y^{3}e^{zy}\nu \left(
dy\right) \geq \left( \int_{a}^{b}y^{2}e^{zy}\nu \left( dy\right) \right)
^{2}. 
\]%
Therefore, the function $H_{z}\left( z,T\right) $ is log-convex and we
conclude using Propositions 13 and 8.

For part (iii), we first observe that for $i,j=0,1,2,...,$ there exist
positive constants $C_{ij}$ such that 
\begin{equation}
\left \vert \frac{\partial }{\partial z}\left( \left( \frac{\partial
^{j}H\left( z,t\right) }{\partial z^{j}}\right) ^{-1}\frac{\partial
^{i}H\left( z,t\right) }{\partial z^{i}}\right) \right \vert \leq C_{ij},
\label{C-inequality}
\end{equation}%
which follows from direct differentiation and repeated application of
inequalities (\ref{H-cm-estimates-t}).

We only establish the first inequality in (\ref{estimates-cm}), say for $%
n=3, $ for the rest follows using similar arguments. To this end, observe
that we have, where all quantities involving $H$ and its derivatives are
evaluated at $(z,t)$, 
\[
H^{2}r_{xxx}\left( H,t\right) =\left( \frac{H}{H_{z}}\right) ^{2}\left(
\left( \frac{H_{zzz}}{H_{z}}\right) _{z}-2\frac{H_{zz}}{H_{z}}\left( \frac{%
H_{zz}}{H_{z}}\right) _{z}\right. 
\]%
\[
\left. -2\frac{H_{zz}}{H_{z}}\left( \frac{H_{zzz}}{H_{z}}-\left( \frac{H_{zz}%
}{H_{z}}\right) ^{2}\right) \right) , 
\]%
and, using (\ref{H-cm-estimates-t}) and (\ref{C-inequality}) repeatedly we
conclude.
\end{proof}

\bigskip

We conclude mentioning that because $I\ $is a completely monotonic function,
Berstein's theorem yields that it can be represented as the Laplace
transform for some positive finite measure $\tilde{\nu}$. Indeed, we have 
\[
I\left( x\right) =\int_{a}^{b}x^{-y}\nu \left( dy\right) =\int_{0}^{\infty
}e^{-x\rho }\tilde{\nu}\left( d\rho \right) , 
\]%
with 
\[
\tilde{\nu}\left( d\rho \right) =\left( \int_{a}^{b}\frac{\rho ^{y-1}}{%
\Gamma \left( y\right) }\nu \left( dy\right) \right) d\rho . 
\]

Complete monotonicity for modeling risk preferences has appeared in \cite%
{brockett}. Therein, this structural property is, however, assumed for the
marginal utility itself, and not for its inverse. Specifically, the authors
consider utility functions with the property 
\begin{equation}
U^{\prime }\left( x\right) =\int_{0}^{\infty }e^{-xy}\mu \left( dy\right) ,
\label{U-CM}
\end{equation}%
for some finite positive measure $\mu $.

Several interesting question arise. Firstly, what are the utility functions
that belong to both classes (\ref{I-cm}) and (\ref{U-CM})? Clearly, power
utilities do as well as some combinations of them. However, this is not the
case for arbitrary sums of power utilities.

The second question is whether complete monotonicity is being preserved at
previous times. Specifically, whether the marginal value function $%
u_{x}\left( x,t\right) $ remains completely monotonic if the marginal
utility $U^{\prime }\left( x\right) $ is as in (\ref{U-CM}). Similarly,
whether $u_{x}^{\left( -1\right) }\left( x,t\right) $ remains completely
monotonic, if the associated $I\left( x\right) $ is as in (\ref{I-cm}).

Thirdly, it is not clear which utilities from the two classes preserve
stochastic dominance of various degrees and, in a different direction, which
utilities allow for a dynamic extension (rolling horizon) of the investment
problem.

The above questions and other issues related to complete monotonicity of the
marginals and their inverses are being currently investigated by the authors
in \cite{kallblad13}.

\section{Appendix}

We discuss the preservation of the log-convexity and log-concavity of
solutions to the heat equation. The log-convexity property is a mere
consequence of H\"{o}lder's inequality while the log-concavity is more
involved. For the latter, we refer the reader to Theorem 1.3 in \cite%
{brascamp75} or to \cite{borell93} and, for the case of boundary data, to 
\cite{keady90}. The one-dimensional case we consider was first proved in 
\cite{schoen}.

\begin{proposition}
\label{nylogconvex} Let $h:\mathbb{D}\rightarrow \mathbb{R}_{+}$ be the
solution of the heat equation 
\[
h_{t}+\frac{1}{2}\left \vert \lambda \right \vert ^{2}h_{xx}=0, 
\]%
with terminal data $h\left( x,T\right) =h_{0}\left( x\right) ,$ with $%
h_{0}\in C^{2}\left( \mathbb{R}\right) $ satisfying $h_{0}\left( x\right) >0$
and the growth assumption $h_{0}\left( x\right) \leq e^{\gamma x},\gamma >0.$
Then, for each $t\in \left[ 0,T\right) ,$ the following assertions hold.

i) If $h_{0}\left( x\right) $ is a log-convex function, then $h\left(
x,t\right) $ is also log-convex.

ii)\ If $h_{0}\left( x\right) $ is a log-concave function, then $h\left(
x,t\right) $ is also log-concave.
\end{proposition}

\begin{proof}
For simplicity, we set $\left \vert \lambda \right \vert ^{2}=1.$

i)\ We need to show that for $\alpha \in \left( 0,1\right) $ and $x,y\in 
\mathbb{R},$ 
\[
h\left( \alpha x+(1-\alpha )y,t\right) \leq h\left( x,t\right) ^{\alpha
}h\left( y,t\right) ^{1-\alpha }. 
\]%
The Feynman-Kac formula, the log-convexity of the terminal datum and H\"{o}%
lder's inequality yield%
\[
h\left( \alpha x+\left( 1-\alpha \right) y,t\right) =E\left( h_{0}\left(
\alpha \left( x+W_{T-t}\right) +\left( 1-\alpha \right) \left(
y+W_{T-t}\right) \right) \right) 
\]%
\[
\leq E\left( \left( h_{0}\left( x+W_{T-t}\right) \right) ^{\alpha }\left(
h_{0}\left( y+W_{T-t}\right) \right) ^{1-\alpha }\right) 
\]%
\[
\leq \left( E\left( h_{0}\left( x+W_{T-t}\right) \right) \right) ^{\alpha
}\left( E\left( h_{0}\left( y+W_{T-t}\right) \right) \right) ^{1-\alpha } 
\]%
\[
=\left( h\left( x,t\right) \right) ^{\alpha }\left( h\left( y,t\right)
\right) ^{1-\alpha }. 
\]%
ii) We need to show that for $\alpha \in \left( 0,1\right) $ and $x,y\in 
\mathbb{R},$ 
\[
h\left( \alpha x+(1-\alpha )y,t\right) \geq h\left( x,t\right) ^{\alpha
}h\left( y,t\right) ^{1-\alpha }. 
\]%
The Pr\'{e}kopa-Leindler inequality yields that if, for $0<\alpha <1,$ $%
z,z^{\prime }\in \mathbb{R},$\ positive functions $f,m,n$ satisfy%
\[
f\left( \alpha z+(1-\alpha )z^{\prime }\right) \geq \left( m\left( z\right)
\right) ^{\alpha }\left( n\left( z^{\prime }\right) \right) ^{1-\alpha }, 
\]%
then, for $z\in \mathbb{R},$%
\[
\int \limits_{-\infty }^{\infty }f\left( z\right) dz\geq \left( \int
\limits_{-\infty }^{\infty }m\left( z\right) dz\right) ^{\alpha }\left( \int
\limits_{-\infty }^{\infty }n\left( z\right) dz\right) ^{1-\alpha }. 
\]%
The log-concavity of $h_{0}\left( x\right) $ yields that for $\alpha \in
\left( 0,1\right) ,$ $z,z^{\prime }\in \mathbb{R},$ 
\[
h_{0}\left( \alpha z+(1-\alpha )z^{\prime }\right) \geq \left( h_{0}\left(
z\right) \right) ^{\alpha }\left( h_{0}\left( z^{\prime }\right) \right)
^{1-\alpha }. 
\]%
Next, fix $(x,y,t)\in \mathbb{R\times R\times }\left[ 0,T\right] $ and
define the functions%
\[
f\left( z;x,y,t\right) :=e^{-\frac{\left( \alpha x+(1-\alpha )y-z\right) ^{2}%
}{4(T-t)}}h_{0}\left( z\right) 
\]%
\[
m\left( z;x,t\right) :=e^{-\frac{\left( x-z\right) ^{2}}{4(T-t)}}h_{0}\left(
z\right) \textrm{ \  \  \  \ and \  \ }n\left( z;y,t\right) :=e^{-\frac{\left(
y-z\right) ^{2}}{4(T-t)}}h_{0}\left( z\right) .\textrm{\ } 
\]%
We easily see, that 
\[
f\left( \alpha z+(1-\alpha )z^{\prime };x,y,t\right) \geq \left( m\left(
z;x,t\right) \right) ^{\alpha }\left( n\left( z^{\prime };y,t\right) \right)
^{1-\alpha }. 
\]%
Indeed, from the log-concavity of the functions $h_{0}\left( x\right) $ and $%
e^{-x^{2}}$ we have%
\[
f\left( \alpha z+(1-\alpha )z^{\prime };x,y,t\right) =e^{-\frac{\left(
\alpha x+(1-\alpha )y-\alpha z-(1-\alpha )z^{\prime }\right) ^{2}}{4(T-t)}%
}h_{0}\left( \alpha z+(1-\alpha )z^{\prime }\right) 
\]%
\[
\geq e^{-\frac{\left( \alpha (x-z)+(1-\alpha )(y-z^{\prime })\right) ^{2}}{%
4(T-t)}}\left( h_{0}\left( z\right) \right) ^{\alpha }\left( h_{0}\left(
z^{\prime }\right) \right) ^{1-\alpha } 
\]%
\[
\geq \left( e^{-\frac{\left( x-z\right) ^{2}}{4(T-t)}}h_{0}\left( z\right)
\right) ^{\alpha }\left( e^{-\frac{\left( y-z^{\prime }\right) ^{2}}{4(T-t)}%
}h_{0}\left( z^{\prime }\right) \right) ^{1-\alpha }. 
\]%
Therefore, 
\[
\int \limits_{-\infty }^{\infty }e^{-\frac{\left( \alpha x+(1-\alpha
)y-z\right) ^{2}}{4(T-t)}}h_{0}\left( z\right) dz 
\]%
\[
\geq \left( \int \limits_{-\infty }^{\infty }e^{-\frac{\left( x-z\right) ^{2}%
}{4(T-t)}}h_{0}\left( z\right) dz\right) ^{\alpha }\left( \int
\limits_{-\infty }^{\infty }e^{-\frac{\left( y-z\right) ^{2}}{4(T-t)}%
}h_{0}\left( z\right) dz\right) ^{1-\alpha }, 
\]%
and we conclude.
\end{proof}


\begin{thebibliography}{99}
\bibitem{agrawal} A. Agarwal and R. Sircar. Portfolio benchmarking under
drawdown constraints and stochastic Sharpe ratio, arXiv:1610.08558v1,
submitted, 2016.

\bibitem{angenant} S. Angenent. The zero set of a solution to a parabolic
equation. \textit{Journal f\"{u}r die reine und angewandte, Mathematik - }%
390, 79-96, 1988.

\bibitem{arrow65} K.J. Arrow. The theory of risk aversion. \textit{Aspects
of the Theory of Risk Bearing, }Yrjo Jahnssosin Saation, Helsinski, 1965;
reprinted in \textit{Essays in the Theory of Risk Bearing, }North Holland,
London, 1970.

\bibitem{bian} B. Bian and H. Zheng. Turnpike property and convergence rate
for an investment model with general utility functions. \textit{Journal of
Economic Dynamics and Control, }51, 28-49, 2015.

\bibitem{bjork09} T.~Bj\"{o}rk. 
\newblock {\em Arbitrage theory in
continuous time}. \newblock Oxford University Press, 2009.

\bibitem{black88} F.~Black. \newblock Individual investment and consumption
under uncertainty. \textit{Portfolio Insurance: A guide to dynamic hedging, }%
D.L. Luskin (ed.). New York, John Wiley and Sons, 1988, 207-225; first
version: November 1, 1968, \textit{Financial Note No. 6B, }Investment and
consumption through time.

\bibitem{bodie} Z. Bodie and J. Treussard. Making investment choices as
simple as possible but not simpler. \textit{Financial Analysts Journal, }%
63(3), 42-47, 2007.

\bibitem{borell} C.~Borell. \newblock Monotonicity properties of optimal
investment strategies for log-{B}rownian asset prices. 
\newblock {\em
Mathematical Finance}, 17(1),143-153, 2007.

\bibitem{borell93} C.~Borell. \newblock Geometric properties of some
familiar diffusions in $R^{n}$. \textit{Annals of Probability}, 21(1),
482-489, 1993.

\bibitem{branch} B. Branch and Q. Liping. Exploring the Pros and Cons of
target-date funds. \textit{Financial Services Review, }20(2), 2011.

\bibitem{brascamp75} H.J. Brascamp and E.H. Lieb. \newblock Some
inequalities for {G}aussian measures. 
\newblock {\em Functional Integral and
Its Applications, A. Arthurs (ed.)}, 1-14, 1975.

\bibitem{brockett} P. Brockett and L. Golden. A class of utilities functions
containing all common utility functions, \textit{Management Science, }33(8),
955-964, 1997.

\bibitem{chen11} A.~Chen, A.~Pelsser, and M.~Vellekoop. \newblock Modeling
non-monotone risk aversion using {SAHARA} utility functions. 
\newblock {\em
Journal of Economic Theory}, 146(5), 2075-2092, 2011.

\bibitem{choulli09} T.~Choulli and M.~Schweizer. The mathematical structure
of horizon-dependence in optimal portfolio choice. Technical report, NCCR
FINRISK working paper No.588, ETH Zurich, 2009.

\bibitem{eechodht} K. Eeckhoudt, C. Gollier and H.S. Schlesinger. \textit{%
Economic financial decisions under risk}, Princeton University Press, 2005.

\bibitem{fouque-hu} J.-P. Fouque and R. Hu. Asymptotic optimal strategy for
portoflio optimization in a slowly varying stochastic environment,
arXiv:1603.03538v2, submitted, 2016.\emph{\ }

\bibitem{fouque} J.-P. Fouque, R. Sircar and T. Zariphopoulou. Portfolio
optimization and stochastic volatility asymptotics, \textit{Mathematical
Finance, }doi.10.1111/mafi.12109.

\bibitem{fukuda} I.~Fukuda, H.~Ishii, and M.~Tsutsumi. \newblock Uniqueness
of solutions to the {C}auchy problem for {$u_{t}-u\Delta u+\gamma
\left
\vert \nabla u\right \vert ^{2}=0$}. 
\newblock {\em Differential and
Integral equations}, 6(6), 1231-1252, 1993.

\bibitem{gollier96} C.~Gollier and J.W. Pratt. \newblock Risk vulnerability
and the tempering effect of background risk. \newblock {\em Econometrica},
64(5), 1109-1123, 1996.

\bibitem{gollier} C.~Gollier and R.J. Zeckhauser. \newblock Horizon length
and portfolio risk. \newblock {\em The Journal of Risk and Uncertainty},
24(3), 195-212, 2002.

\bibitem{Gollier-book} C.~Gollier. \newblock \textit{The economics of risk
and time}. MIT Press, 2001.

\bibitem{guiso08} L.~Guiso and M.~Paiella. \newblock Risk aversion, wealth,
and background risk. 
\newblock {\em Journal of the European Economic
Association}, 6(6), 1109-1150, 2008.

\bibitem{he} H.~He and C.~F. Huang. \newblock Consumption-portfolio
policies: An inverse optimal problem. 
\newblock {\em Journal of Economic
Theory}, 62(2), 257-293, 1994.

\bibitem{hennessy06} D.A. Hennessy and H.E. Lapan. \newblock On the nature
of certainty equivalent functionals. 
\newblock {\em Journal of Mathematical
Economics}, 43(1), 1-10, 2006.

\bibitem{huang99} C.F. Huang and T.~Zariphopoulou. \newblock Turnpike
behavior of long-term investments. \newblock {\em Finance and Stochastics},
3(1), 1-20, 1999.

\bibitem{kallblad-thesis} S. K\"{a}llblad. Topics in portfolio choice:\
qualitative properties, time-consistency and investment under model
uncertainty, D.Phil. Thesis, University of Oxford, 2014.

\bibitem{kallblad-Zar} S. K\"{a}llblad and T. Zariphopoulou. Qualitative
analysis of optimal investment strategies in log-normal models,
hhtps://ssrn.com/2373587, 2014.

\bibitem{kallblad13} S.~K\"{a}llblad and T.~Zariphopoulou. \newblock %
Structural representation of utilities and their effects on horizon
flexibility and stochastic dominance. {P}reprint, \newblock2016.

\bibitem{karatzas87} I.~Karatzas, J.P. Lehoczky, and S.E. Shreve. \newblock %
Optimal portfolio and consumption decisions for a small investor on a finite
horizon. \newblock {\em SIAM Journal on Control and Optimization.}, 25(6),
1557-1586, 1987.

\bibitem{keady90} G.~Keady. \newblock The persistence of logconcavity for
positive solutions of the one dimensional heat equation. 
\newblock {\em J.
Austral. Math. Soc. Ser. A}, 48, 1-16, 1990.

\bibitem{Kimball-1} M.S. Kimball. \newblock Precautionary savings in the
small and in the large. \newblock \emph{Econometrica}, 58, 53-73, 1990.

\bibitem{kimball} M.S. Kimball. Standard risk aversion. \textit{%
Econometrica, }61(3), 589-611, 1993.

\bibitem{lajeri} F.~Lajeri and L.T. Nielsen. \newblock Parametric
characterizations of risk aversion and prudence. 
\newblock {\em Economic
Theory}, 15(2), 469-476, 2000.

\bibitem{larsen12} K.~Larsen and H.~Yu. \newblock Horizon dependence of
utility optimizers in incomplete models. 
\newblock {\em Finance and
Stochastics}, 16(4), 779-801, 2012.

\bibitem{lintner} J. Lintner. \newblock The valuation of risky assets and
the selection of risky investments in stock portfolio and capital budgets. %
\newblock \textit{Review of Economics and Statistics,}\emph{\ }47(1), 13-37,
1965.

\bibitem{lorig} M. Lorig. Indifference prices and implied volatility. 
\textit{Mathematical Finance,} doi:10.111/mafi.12129.2016.

\bibitem{lorig-sircar} M. Lorig and R. Sircar. Portfolio optimization under
local-stochastic volatility, coefficient Taylor series approximations and
implied Sharpe ratio, \textit{SIAM Journal on Financial Mathematics, }7(1),
418-447, 2016.

\bibitem{maggi} M.A. Maggi, U. Magnani and M. Menegatti. On the relationship
between absolute prudence and absolute risk aversion. \textit{Decisions in
Economics and Finance, }29(2),155-160, 2006.

\bibitem{merton} R.C. Merton. \newblock Lifetime portfolio selection under
uncertainty: the continuous time case. 
\newblock {\em The Review of
Economics and Statistics}, 51, 247--257, 1969.

\bibitem{mossin} J. Mossin. \newblock Optimal multi-period portfolio
policies. \newblock \textit{Journal of Business,}\emph{\ }41(2):215-229,
1970.

\bibitem{musiela-SIFIN} M. Musiela and T. Zariphopoulou. Portfolio choice
under space-time monotone performance criteria. \textit{SIAM\ Journal on
Financial Mathematics, }1, 326-365, 2010.

\bibitem{pratt} J.W. Pratt. \newblock Risk aversion in the small and in the
large. \newblock {\em Econometrica}, 32(1), 122-136, 1964.

\bibitem{zeck} J.W. Pratt and R.J. Zeckhauser. \newblock Proper risk
aversion. \newblock {\em Econometrica}, 55(1), 143-154, 1987.

\bibitem{protter} M.H. Protter and H.F. Weinberger. 
\newblock {\em
Maximum principles in differential equations}. \newblock Springer Verlag,
New York, 1984.

\bibitem{R-S1} M. Rothschild and J. Stiglitz. \newblock Increasing Risk: I.
A definition. \newblock \textit{Journal of Economic Theory,}\emph{\ }2,
225-243, 1970.

\bibitem{R-S2} M. Rothschild and J. Stiglitz. \newblock Increasing Risk: II.
Its economomic consequences. \newblock \textit{Journal of Economic Theory,}%
\emph{\ }3, 225-243, 1971.

\bibitem{samuelson89} P.A. Samuelson. \newblock The judgment of economic
science on rational portfolio management: indexing, timing, and long-horizon
effects. \newblock {\em The Journal of Portfolio Management}, 16(1), 4-12,
1989.

\bibitem{schoen} I.J. Schoenberg. On Polya frequency functions I: The
totally positive fucntions and their Laplace transforms. \textit{Journal d'
Analyse Math\'{e}matique,} 1, 331-374, 1951.

\bibitem{Shkolnikov} Shkolnikov, M.: On a nonlinear transformation between
Brownian martingales, arXiv:1205.3218v1, 2012.

\bibitem{spitzer} J. Spitzer and S. Sandeep. Target-date mutual funds. 
\textit{Consumer knowldge and financial decisions, }D.J. Lamdin (ed.),
Springer, 2012.

\bibitem{surz} R.J. Surz and C.L. Israelsen. Evaluating target date
lifecycle funds. \textit{The Journal of Portfolio Management, }12, 62-70,
2007.

\bibitem{book} D.V. Widder. \newblock {\em The Heat Equation}. \newblock %
Academic Press, 1975.

\bibitem{Xia} J.~Xia. \newblock Risk aversion and portfolio selection in a
continuous-time model. \textit{SIAM Journal on Control and Optimization},
49(5), 1916-1937, 2011.

\bibitem{vasquez} J.-L. Vasquez. \textit{The porous medium equation}. Oxford
University Press, 2007.

\bibitem{zariphopoulou09} T.~Zariphopoulou and T.~Zhou. \newblock Investment
performance measurement under asymptotically linear local risk tolerance. %
\newblock {\em Handbook of Numerical Analysis, P.G, Ciarlet (ed.)}, 15,
227-253, 2009.
\end{thebibliography}
\end{document}